\documentclass[conference]{IEEEtran}
\usepackage{amsmath,amsfonts}
\usepackage{algorithmic}
\usepackage{algorithm}
\usepackage{array}
\usepackage{textcomp}
\usepackage{url}
\usepackage{verbatim}
\usepackage{graphicx,flushend}
\usepackage{cite,color,tcolorbox}
\usepackage{amssymb,setspace}
\usepackage{amsthm}
\usepackage{subcaption} 
\newtheorem{proposition}{Proposition}

\usepackage{hyperref}
\usepackage{geometry}
\begin{document}
\title{Joint Communication and Trajectory Design for Movable Antenna Systems}
\IEEEoverridecommandlockouts
\author{\IEEEauthorblockN{Jiaxuan Li\IEEEauthorrefmark{1}, Weidong Mei\IEEEauthorrefmark{1}, Changhao Liu\IEEEauthorrefmark{1}, Zhi Chen\IEEEauthorrefmark{1}, Boyu Ning\IEEEauthorrefmark{1}, and Rui Zhang\IEEEauthorrefmark{2}}
\IEEEauthorblockA{\IEEEauthorrefmark{1}National Key Laboratory of Wireless Communications,\\ University of Electronic Science and Technology of China, Chengdu, China  611731}
\IEEEauthorblockA{\IEEEauthorrefmark{2}Department of Electrical and Computer Engineering, National University of Singapore, Singapore 119077}
Emails: jiaxuanli@std.uestc.edu.cn, wmei@uestc.edu.cn, lch3001@163.com, chenzhi@uestc.edu.cn,\\ boydning@outlook.com, elezhang@nus.edu.sg}
\maketitle

\begin{abstract}
Movable antennas (MAs) have attracted significant attention in wireless communications due to their ability to reconfigure channel conditions by flexibly adjusting the antenna positions within a confined region. However, MA movement generally incurs a non-negligible delay, which may significantly limit the data transmission time at optimized positions. To tackle this challenge, this paper investigates a new joint communication and trajectory optimization problem, where each MA transmits while moving along an optimized trajectory to prolong the effective data transmission time. Focusing on a single-MA system, our goal is to maximize the average data rate by optimizing the MA's positions over time, subject to its maximum velocity constraints. However, this continuous-time antenna position optimization problem is highly non-convex and challenging to solve. To tackle this challenge, we first consider a special case with two channel paths and derive the optimal MA trajectory in closed form. For other general cases, we ingeniously reformulate the average rate maximization problem into a fixed-hop shortest path problem in graph theory by sampling the antenna movement region into a multitude of discrete points, and solve it optimally. Simulation results demonstrate that our proposed algorithm can significantly improve the data rate compared to other baseline schemes.
\end{abstract}

\section{Introduction}
Existing communication systems mainly rely on fixed-position antennas (FPAs). However, due to small-scale fading in practical multi-path channels, FPAs may be trapped at unfavorable locations with deep fading, which fundamentally limits their communication performance. Furthermore, from an array signal processing perspective, FPAs result in fixed spatial correlation among array responses corresponding to different angles, thus limiting the flexibility of beamforming designs.

To tackle these limitations, movable antennas (MAs) (and their equivalents, such as fluid antennas (FAs)) have emerged as promising solutions. Compared to FPAs, MAs can fully leverage the spatial degrees of freedom (DoFs) through dynamically adjusting their positions within a confined movement region at the transmitter/receiver, thereby circumventing deep-fading positions and enhancing beamforming flexibility \cite{1, 2, zhu2026matot, li2025marelay}. Thanks to these unique advantages, MA systems and the associated antenna position optimization problems have been widely investigated in a wide range of scenarios, including physical layer security \cite{3}, integrated sensing and communication (ISAC) \cite{5}, near-field communications \cite{6}, wide-beam coverage \cite{7}, mobile edge computing \cite{CPCMEC2024}, intelligent reflecting surface (IRS)-aided communications \cite{WXIRS2025}, unmanned aerial vehicle (UAV)-enabled communications \cite{LCH6DMA2026}, cognitive radio \cite{wei2024spectrumsharing,zhang2026cognitive}, and so on. MA systems with joint antenna positioning and rotation optimization have also been studied in \cite{shao20256dma, shao20256dmamag, xie20266dma}.
\begin{figure}[!t]
	\centering
	\includegraphics[width=3.2in]{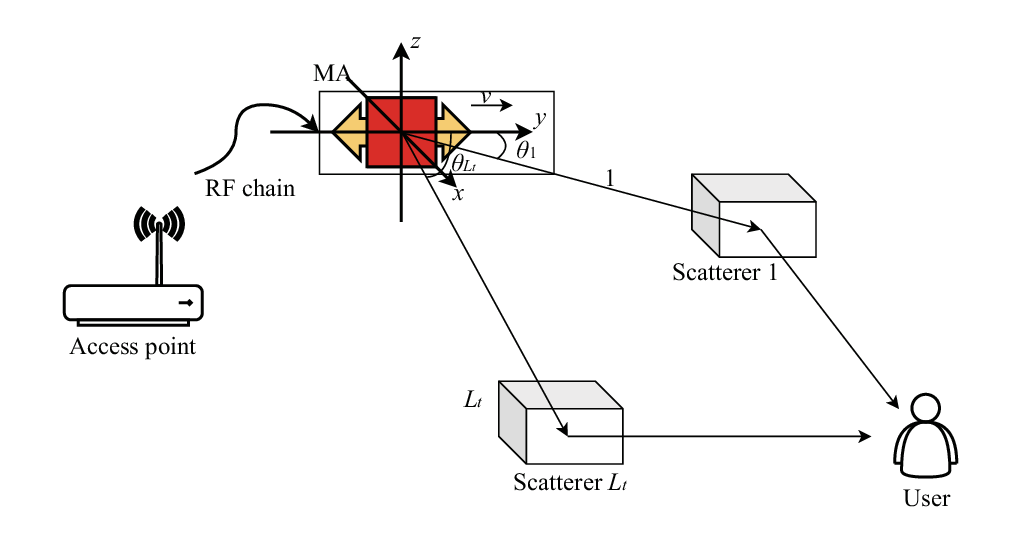}
	\caption{MA-SISO system with joint communication and antenna movement.}
	\label{fig.1}
    \vspace{-9pt}
\end{figure}

Despite these advances, existing works mostly optimize the MA system performance by simply determining a set of optimal destination antenna positions within the movement region. All antennas are then mechanically relocated to these positions for data transmission \cite{1, 2}. However, given the large mechanical movement delay in practice, this strategy may severely limit the transmission time available at the destination positions. Moreover, it overlooks the communication performance during the antenna movement phase, which could contribute more significantly to the overall data rate. In light of these facts, this paper investigates a new joint communication and trajectory optimization problem, where each MA transmits while moving along an optimized trajectory to improve data transmission quality throughout the entire communication period. It is worth noting that MA trajectory designs have been studied in some recent works. For example, in \cite{10,19}, the authors optimized MA movement trajectory to minimize movement delay and maximize the energy efficiency (by accounting for mechanical power consumption), respectively. In particular, both works revealed that for one-dimensional (1D) antenna movement, the optimal trajectory admits a closed-form solution, with each MA moving at the maximum velocity. However, both works still follow the conventional movement strategy as in prior works, without accounting for the communication performance during the movement phase.

Motivated by the above gap, this paper develops a joint communication and trajectory optimization framework for a single-MA system, as shown in Fig. \ref{fig.1}. To characterize its optimal performance, we aim to maximize the average data rate at the receiver by optimizing the MA's positions over time subject to its maximum velocity constraint. However, such a continuous-time antenna position optimization problem turns out to be much more challenging to solve compared with the one-time antenna position optimization problems as considered in prior works \cite{3,5,6,7,CPCMEC2024,WXIRS2025,LCH6DMA2026}. To tackle this challenge, we first consider a special case with two channel paths only and derive the optimal antenna trajectory in closed-form. We then consider the general multi-path case and ingeniously reformulate the original problem as a fixed-hop shortest path problem by sampling the movement region into a multitude of discrete points. Simulation results demonstrate that the proposed graph-based algorithm can achieve a higher rate performance than other benchmark schemes by properly balancing the communication performance and movement time.\vspace{-3pt}

\begingroup
\allowdisplaybreaks
\section{System Model and Problem Formulation}
\subsection{System Model}
\begin{figure}[!t]
	\centering
	\includegraphics[width=3in]{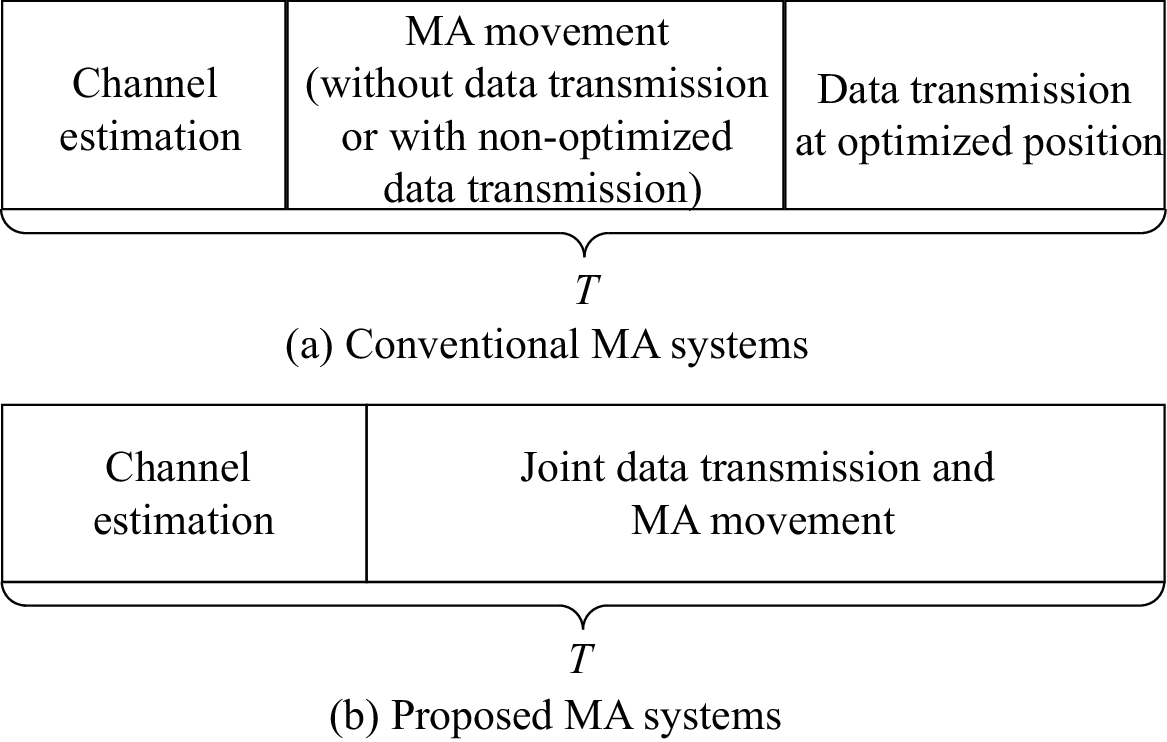}
	\caption{Transmission protocol of the proposed MA system.}
	\label{fig.2}
    \vspace{-9pt}
\end{figure}
As depicted in Fig. \ref{fig.1}, we consider a single-input single-output (SISO) MA system consisting of a single-MA transmitter (e.g., an access point (AP)) and a single-FPA receiver. Driven by a stepper motor, the MA can be flexibly positioned within a linear array of length $D$, denoted as $\mathcal{C}_t$. For ease of exposition, we assume a block-fading channel model and denote by $T$ the duration of each block, during which the channel can be viewed as approximately constant. At the beginning of the block, the transmitter estimates the channel at each position within $\mathcal{C}_t$ by leveraging existing channel estimation technologies dedicated to MAs (see e.g., \cite{ma2023compressed,xiao2024compressed,zhang2025successive,huang2025cnn}). In contrast to the existing works that neglect communication performance during MA movement, as illustrated in Fig.\;\ref{fig.2}(a), we propose a joint optimization framework as illustrated in Fig.\;\ref{fig.2}(b), where the MA follows an optimized trajectory while simultaneously transmitting data to the receiver based on the channel estimates.

Let $x_0 \in {\cal C}_t$ denote the initial position of the MA. To facilitate the MA trajectory design, we discretize the time duration $T$ into $K$ time slots each with a length of $\tau=\frac{T}{K}$, where $K$ is assumed to be sufficiently large such that the channel remains approximately static even when the MA moves at the maximum velocity. Let $v\left[k\right]$ and $x\left[k\right]$ denote the MA's velocity and position in time slot $k$, respectively. It follows that
\begin{equation}
	x\left[k\right] = x\left[k-1\right] + v\left[k\right]\tau, k\in \mathcal{K}\triangleq \{1, 2, \ldots, K\},\label{eq1}
\end{equation}
and the MA's trajectory can be denoted as a sequence $\left\lbrace x\left[0\right], x\left[1\right], \ldots, x\left[K\right]\right\rbrace$, with $x[0]=x_0$. Moreover, due to the finite power of practical stepper motors, the MA has a maximum moving speed, denoted as $v_{\max}$, which leads to
\begin{equation}
	 \left| x\left[k\right] - x\left[k-1\right]\right| \leq v_{\max}\tau, k\in \mathcal{K}. \label{eq2}
\end{equation}

In this paper, we adopt the field-response channel model to characterize the channel between the MA and the receiver \cite{11}. Let $L_t$ denote the number of multi-path components and $\theta_l$ denote the angle of departure (AoD) of the $l$-th path, $1 \leq l \leq L_t$.
Accordingly, for the $l$-th path, the phase-shift difference of the MA in time slot $k$ relative to the reference point is given by $\Phi_l\left[k \right] =\frac{2\pi}{\lambda} x\left[k\right]\cos{\theta_l}$. The field-response vector (FRV) in time slot $k$ is given by
\begin{equation}
 \mathbf{g}\left[k\right]  = \left[e^{j\Phi_1\left[k \right]}, e^{j\Phi_2\left[k \right]}, \ldots, e^{j\Phi_{L_t}\left[k \right]} \right]^T.\label{eq7}
\end{equation}
Consequently, the overall channel in time slot $k$ is given by
\begin{equation}
	{h}\left[ k\right] = \mathbf{a}^H\mathbf{{g}}[k],\label{eq8}
\end{equation}
where $\mathbf{a} = \left[a_1, a_2, \ldots, a_{L_t} \right]^T$ denotes the multi-path channel response vector, and $a_l$ denotes the complex coefficient of the $l$-th path. The received signal at the receiver in time slot $k$ is given by
\begin{equation}
	y\left[ k\right]= {h}\left[ k\right] \sqrt{P_t} s\left[k\right]  + n\left[k\right],\label{eq9}
\end{equation}
where $s\left[k\right]$ and $n\left[k\right]$ denote the transmitted symbol with a unit power and the additive white Gaussian noise (AWGN) at the receiver in time slot $k$, respectively. Let $P_t$ denote the transmit power in each time slot. The achievable rate at the receiver in time slot $k$ is thus given by
\begin{equation}
	R\left[k\right] = \log_2\left(1 + \frac{P_t\left|{h}\left[ k\right]\right| ^2}{\sigma^2} \right), k \in {\cal K},\label{eq6}
\end{equation}
where $\sigma^2$ denotes the average noise power at the receiver.\vspace{-3pt}

\subsection{Problem Formulation}
This paper aims to maximize the average achievable rate at the receiver over the transmission duration $T$, i.e., $\frac{1}{K}\sum^K_{k=1} R\left[k\right]$, by optimizing the MA trajectory $x[k], k \in {\mathcal{K}}$. 
The associated optimization problem can be formulated as
\begin{align}
    \left(P1 \right):\ \  &\underset{x\left[k\right], k\in\mathcal{K}}{\max}\ \  \frac{1}{K}\sum^K_{k=1} R\left[k\right] \label{eq:11A}\\
    \mbox{s.t.}\ &\left|x\left[k \right] - x\left[k-1 \right]\right| \leq v_{\max}\tau, k\in\mathcal{K}, \label{eq:11B}\\
    &\ x_0 = x\left[0\right],  \label{eq:11C}\\
    &\ x\left[k\right]\in\mathcal{C}_t, k\in\mathcal{K}. \label{eq:11D}
\end{align}
However, $\left(P1 \right)$ is a non-convex optimization problem due to the highly nonlinear channel responses with respect to (w.r.t.) the antenna position. Moreover, as the MA's positions in different time slots are coupled due to the constraints in \eqref{eq:11B}, existing antenna position optimization algorithms become ineffective and cannot be directly applied. In the following section, a graph-based approach will be proposed to obtain an optimal solution to $\left(P1 \right)$ by transforming it into a discrete problem.

\section{Optimal Solution to $(P1)$}
In this section, we first show that a closed-form optimal solution to $(P1)$ can be derived under $L_t=2$, followed by an optimal graph-based approach for other general cases.\vspace{-2pt}

\subsection{Two-Path Case}
First, we assume that there are only $L_t=2$ channel paths between the transmitter and the receiver. In this case, the channel power gain in time slot $k$ can be simplified as
\begin{equation}
	\left|h[k]\right|^2 =  \left|a_1 \right|^2 + \left|a_2\right|^2 + 2\left|a_1 \right|\left|a_2 \right|\cos\varphi[k],\label{eq12}
\end{equation}
where $\varphi[k] = \frac{2\pi}{\lambda}x\left[k\right]\left(\cos\theta_1-\cos\theta_2 \right)+\Delta\phi$, with $\Delta\phi=\angle{a_2}-\angle{a_1}$. Notably, the superimposed channel power gain of the two paths in \eqref{eq12} demonstrates a periodic character w.r.t. the antenna position. Particularly, it attains the maximum at $\cos\varphi[k] = 1$, corresponding to $x[k]=\frac{\left(2n\pi-\Delta\phi\right) \lambda}{2\pi\left(\cos\theta_1-\cos\theta_2 \right)}, n\in\mathbb{Z}$, where $\mathbb{Z}$ denotes the set of integers.

Accordingly, we define a set of coherent positions within ${\cal C}_t$, which render the two paths coherently combined at the receiver, i.e., $\mathcal{X}=\left\lbrace \hat{x}| \hat{x} = \frac{\left(2n\pi-\Delta\phi\right) \lambda}{2\pi\left(\cos\theta_1-\cos\theta_2 \right)}, n\in\mathbb{Z}, \hat{x} \in\mathcal{C}_t\right\rbrace$. Based on this, we provide the following proposition to obtain the optimal solution to (P1) for $L_t=2$.
\begin{proposition}
	In the case of $L_t=2$, let $x^* = \arg\underset{\hat{x}\in\mathcal{X}}{\min}\left|\hat{x}-x_0\right|$ denote the closest coherent position to $x_0$. The optimal solution to $(P1)$ is given by
	\begin{equation}
		x\left[k\right]=x_0+\frac{x^*-x_0}{\left|x^*-x_0\right| }v_{\max} \tau \min\left\lbrace \frac{\lvert x^*-x_0 \rvert}{v_{\max}\tau},K\right\rbrace , k \in {\mathcal{K}}.\label{eq01}
	\end{equation}
\end{proposition}
\begin{proof}
	Please refer to the Appendix.
\end{proof}	
\textbf{Proposition 1} indicates that in the case of $L_t=2$, the optimal MA trajectory is achieved by moving the MA at the maximum speed towards $x^*$. When the time duration is sufficiently long, i.e., $v_{\max}\tau K \ge \lvert x^*-x_0 \rvert$, the MA will stop at $x^*$ for the remaining time. This is expected, since for any initial position $x_0$, moving towards $x^*$ always results in an increasing channel power gain with the receiver, and moving at the maximum speed helps attain the maximum channel power gain more quickly.

\subsection{Multi-Path Case} 
For other general cases with $L_t \geq 3$, we propose a graph-based method in this subsection. First, we discretize the antenna movement region ${\cal C}_t$ into $N$ grids, with an equal spacing between the centers of any two adjacent grids given by $\delta_s = D/N$. The position of the $n$-th grid center is thus given by $p_n = nD/N$. Note that a sufficiently large value of $N$ should be adopted to ensure the channel over each grid is approximately constant. Let $\mathcal N \triangleq \{1,2,\cdots,N\}$ denote the set of all grid centers. Thus, the position of the MA can be approximated as one of the grid centers in $\mathcal N$, and the initial position $x_0$ can be approximated as the $s$-th grid center, with $s=x_0/\delta_s$, which is assumed to be an integer. Next, we construct a graph $G = \left(V, E\right)$, where the vertex set $V$ is defined as the set of the grids, i.e., $V = \mathcal{N}$. Moreover, we add an edge between two vertices $i$ and $j$ if and only if $\lvert i - j \rvert \leq d_{\max} \triangleq v_{\max}\tau/\delta_s > 1$, i.e., the $j$-th grid can be reached from the $i$-th grid (and vice versa) within a time slot. Hence, the edge set can be constructed as
\begin{equation}
	E = \left\lbrace(i,j)\mid\left| i-j \right| \leq d_{\max}, i, j\in V \right\rbrace. \label{eq13}
\end{equation}
It should be mentioned that each vertex has a self-loop (corresponding to stopping at this position) based on \eqref{eq13}, making $G$ cyclic.

Following the above, one ``dummy" vertex $N+1$ is added to $G$, and an edge is added from each vertex in $V$ to the vertex $N+1$, leading to new vertex and edge sets given by
\begin{align}
	\tilde{V} = V \cup \left\lbrace N+1\right\rbrace, \;
	\tilde{E} = E \cup \left\lbrace(j, N+1)\mid j \in V \right\rbrace, \label{eq14}
\end{align}
respectively. Given the new graph $\tilde{G}=(\tilde{V}, \tilde{E})$, we assign each edge $(i,j)\in\tilde{E}$ with a weight given by
\begin{equation}
    {W}_{i,j}=-\log_2\left(1 + \frac{P_t\left|h\left(p_i\right)  \right| ^2}{\sigma^2} \right), \forall(i,j)\in\tilde{E},
\label{eq15}
\end{equation} 
where $h\left(p_i\right)$ denotes the channel with the user when the MA is positioned at the $i$-th grid center. Based on the above, it can be shown that $\left(P1\right)$ is equivalent to the following fixed-hop shortest path problem $\left(P2\right)$.
\begin{tcolorbox}[standard jigsaw, opacityback=0]
$\left(P2\right):$ \textit{Find the $(K+1)$-hop shortest path from vertex $s$ to vertex $N+1$ in graph $\tilde{G}$.}
\end{tcolorbox}

Note that $(P2)$ is a fixed-hop shortest path problem, for which the optimal solution can be obtained by invoking dynamic programming (DP). The details of DP can be found in \cite{23}, where a similar fixed-hop shortest path problem is formulated and solved in a recursive manner. It can be shown that the DP incurs a complexity order of $\mathcal{O}(\lvert \tilde V \rvert \lvert \tilde E \rvert) = \mathcal{O}\left(KN^2\right)$.

\textbf{Remark 1.} \textit{Note that although this paper only considers a point-to-point scenario, the proposed graph-based algorithm is also applicable to a wide range of scenarios, as long as the objective function can be expressed as the sum of a given utility over time. For example, consider the presence of a single-antenna eavesdropper in our system. If our goal is to maximize the average secrecy rate, we can simply redefine the weight $W_{i,j}$ as the negative of the secrecy rate when the antenna is positioned at the $i$-th grid center, and then invoke DP to solve the corresponding fixed-hop shortest path problem.}
\begin{figure}[!t]
	\centering
	\includegraphics[width=3in]{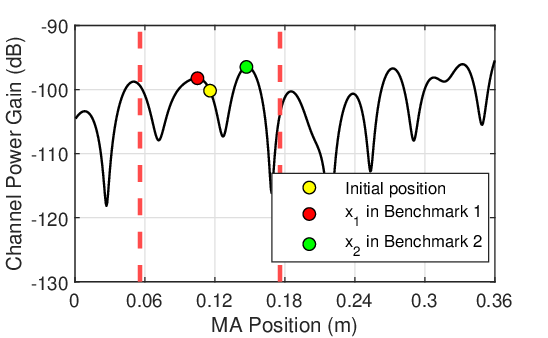}
	\caption{Illustration of the destination positions in the myopic- and far-sighted schemes.}
	\label{fig.9}
    \vspace{-6pt}
\end{figure}
\begin{figure*}[hbtp]
	\centering
	\begin{subfigure}{0.32\textwidth}
		\centering
		\includegraphics[width=\linewidth]{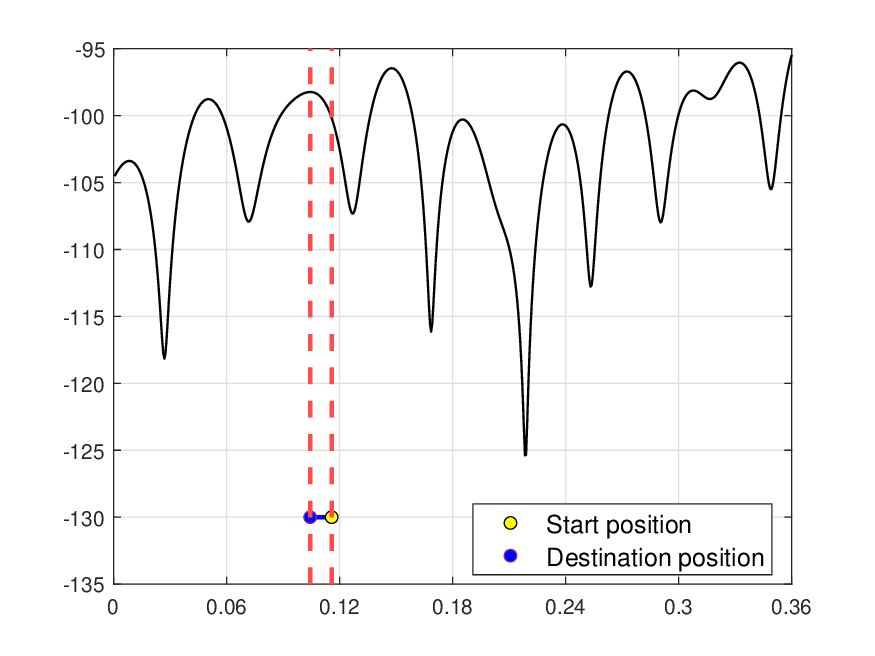} 
		\caption{Graph-based, $T=0.5$ s.}
		\label{fig:sub11}
	\end{subfigure}
	\hfill
	\begin{subfigure}{0.32\textwidth}
		\centering
		\includegraphics[width=\linewidth]{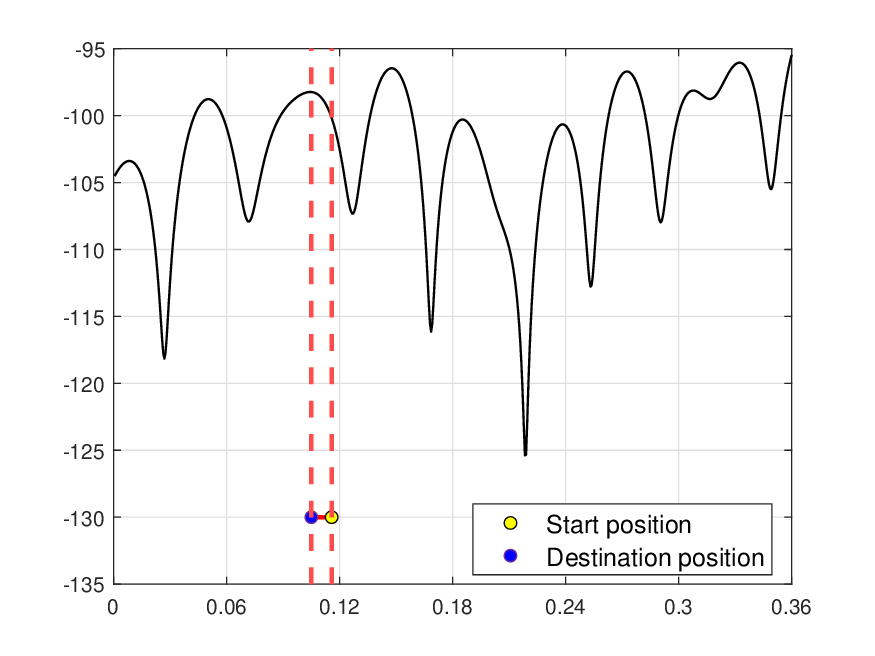} 
		\caption{Myopic, $T=0.5$ s.}
		\label{fig:sub12}
	\end{subfigure}
	\hfill
	\begin{subfigure}{0.32\textwidth}
		\centering
		\includegraphics[width=\linewidth]{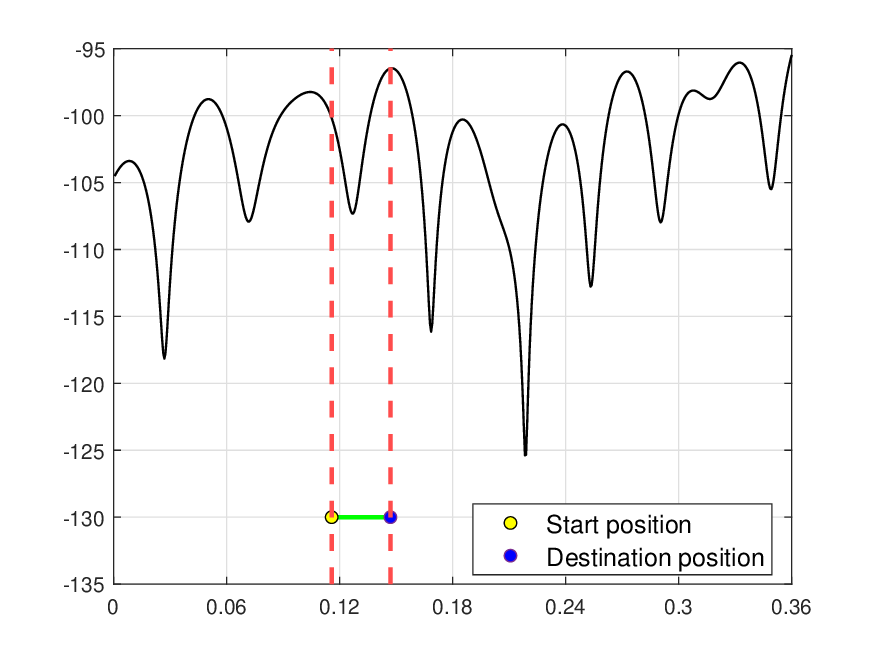} 
		\caption{Far-sighted, $T=0.5$ s.}
		\label{fig:sub13}
	\end{subfigure}
	\vspace{0.1em}
	
	\begin{subfigure}{0.32\textwidth}
		\centering
		\includegraphics[width=\linewidth]{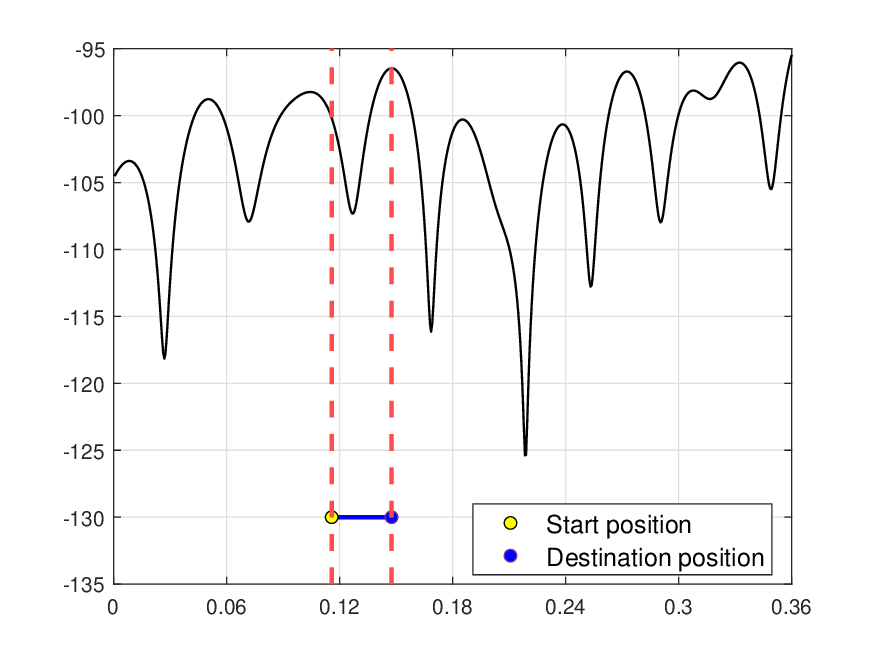}
		\caption{Graph-based, $T=1$ s.}
		\label{fig:sub21}
	\end{subfigure}
	\hfill
	\begin{subfigure}{0.32\textwidth}
		\centering
		\includegraphics[width=\linewidth]{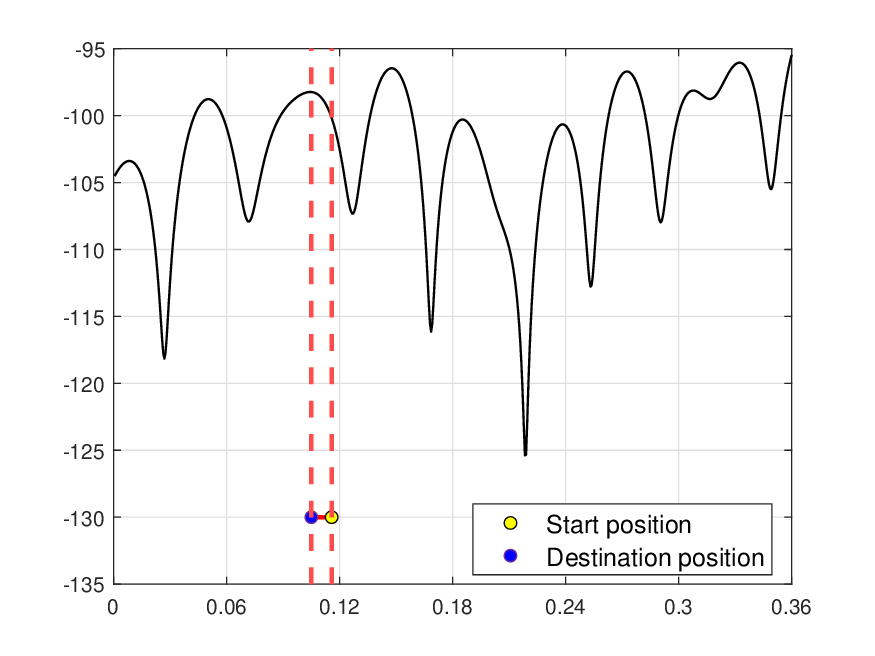}
		\caption{Myopic, $T=1$ s.}
		\label{fig:sub22}
	\end{subfigure}
	\hfill
	\begin{subfigure}{0.32\textwidth}
		\centering
		\includegraphics[width=\linewidth]{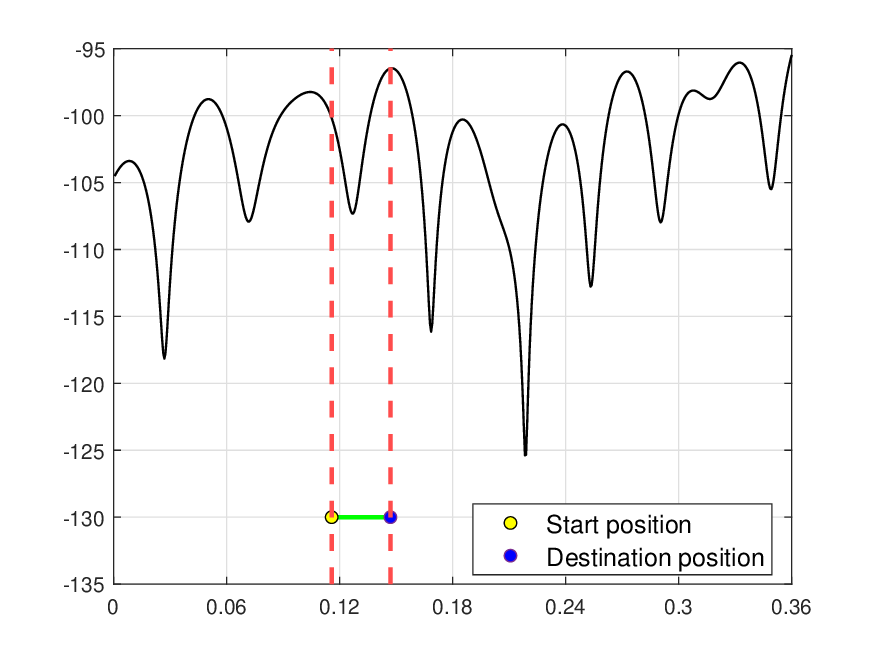}
		\caption{Far-sighted, $T=1$ s.}
		\label{fig:sub23}
	\end{subfigure}
	\vspace{0.1em}
	
	\begin{subfigure}{0.32\textwidth}
		\centering
		\includegraphics[width=\linewidth]{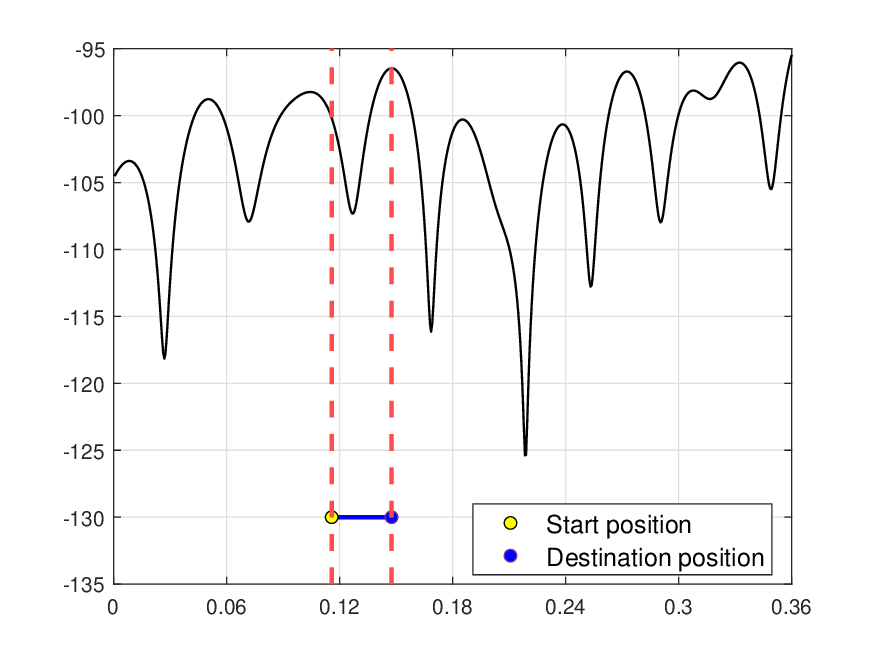}
		\caption{Graph-based, $T=2$ s.}
		\label{fig:sub31}
	\end{subfigure}
	\hfill
	\begin{subfigure}{0.32\textwidth}
		\centering
		\includegraphics[width=\linewidth]{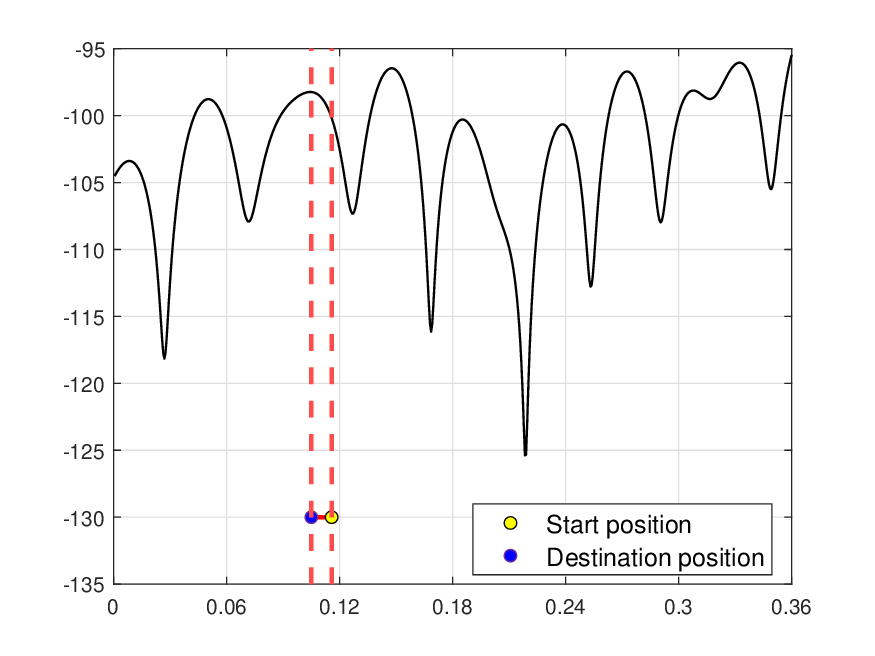}
		\caption{Myopic, $T=2$ s.}
		\label{fig:sub32}
	\end{subfigure}
	\hfill
	\begin{subfigure}{0.32\textwidth}
		\centering
		\includegraphics[width=\linewidth]{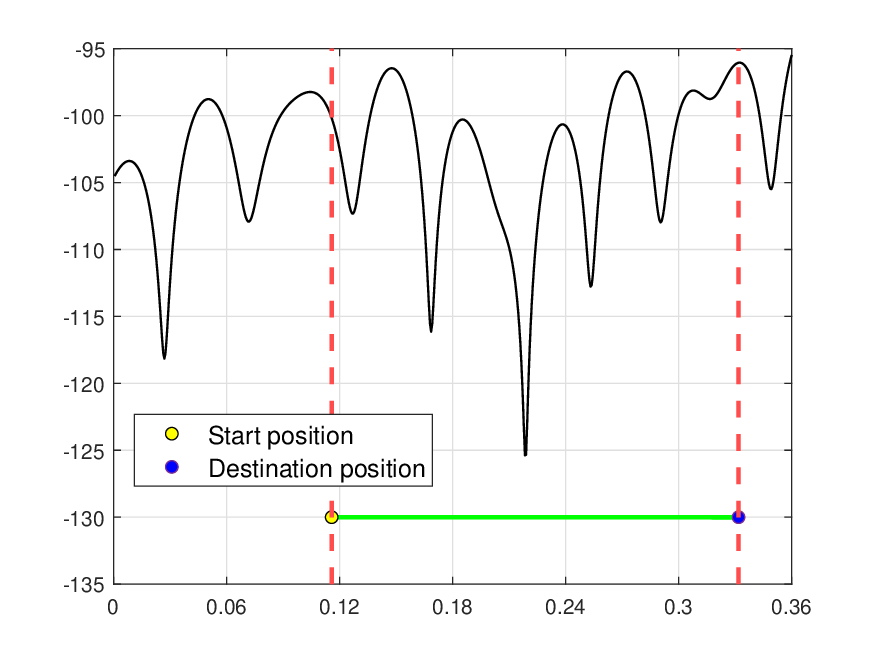}
		\caption{Far-sighted, $T=2$ s.}
		\label{fig:sub33}
	\end{subfigure}
	\caption{Trajectories of different schemes under different time durations $T$.}
	\label{fig:main_3x3}
    \vspace{-6pt}
\end{figure*}

\section{Simulation Results}
In this section, we provide numerical results to evaluate the performance of our proposed graph-based algorithm. Unless otherwise stated, the simulation parameters are set as follows. The carrier wavelength is $\lambda = 0.06$ m, and the length of the linear array is $D = 6\lambda = 0.36$ m. The initial position of the MA is randomly generated within $\mathcal{C}_t$ and its maximum velocity is $v_{\max}=2\lambda=0.12$ m/s. The duration of each time slot is set to $0.01$ s, i.e., $K=100T$. The number of grids after the discretization is $N=600$; thus, we have $\delta_s=\lambda/100$ and $d_{\max}=2$. The AoDs are randomly generated, each following a uniform distribution within $\left[0, \pi\right]$. The distance between the transmitter and the receiver is $100$ m and the path loss exponent is set to $2.8$. The transmit power is $P_t=40$ W and the average noise power at the receiver is $\sigma^2=10^{-11}$ W. All results are averaged over $1000$ channel realizations. 
\par
To validate the performance of the proposed graph-based algorithm, we consider the following three benchmark schemes for performance comparison. 
\begin{itemize}
	\item \textbf{Benchmark 1 (Myopic scheme):} Due to the random superposition of multi-path components, the channel power gain within $\mathcal{C}_t$ exhibits multiple local maxima or crests. Let $\mathcal{X}_1$ denote the set of all crests in $\mathcal{C}_t$ and $x_1=\arg\underset{ x\in\mathcal{X}_1}{\min}|x-x_0|$ denote the closest crest to $x_0$. In this scheme, the MA moves toward $x_1$ with the maximum speed and then stops at $x_1$ if $v_{\max}T \geq \lvert x_1 - x_0 \rvert$.
	\item \textbf{Benchmark 2 (Far-sighted scheme):} Let $x_2$ denote the position with the maximum channel power gain among all positions that are reachable from $x_0$ within the time duration $T$. In this scheme, the MA moves toward $x_2$ with the maximum speed and then stops at $x_2$ if $v_{\max}T \geq \lvert x_2-x_0 \rvert$.
	\item \textbf{Benchmark 3 (FPA):}  In this scheme, the antenna is fixed at the center of $\mathcal{C}_t$.
\end{itemize}
An illustration of $x_1$ and $x_2$ in Benchmarks 1 and 2 is shown in Fig. \ref{fig.9}, marked by the red and green points, respectively. The interval between the two red dashed lines represents the set of positions that can be reached from the initial position $x_0$ within the time duration $T$.

Fig. \ref{fig:main_3x3} shows the MA trajectories by different schemes under different values of the time duration $T$. It is observed that when $T=0.5$ s, the destination positions for both our proposed scheme and the myopic scheme remain the same, while that for the far-sighted scheme is farther from the initial position. As a result, the MA in the far-sighted scheme may traverse a greater number of deep-fading positions or valleys during its movement. When $T=1$ s, however, the destination position for our proposed scheme differs from that for the myopic scheme but remains the same as that for the far-sighted scheme, despite a valley in between. This is attributed to the longer time duration $T$, which allows the MA to enjoy the high channel power gain at a more distant destination position for a longer time. In contrast, for $T=2$ s, the destination position for our proposed scheme differs from both benchmarks and lies between their respective destination positions. This suggests that the proposed scheme can achieve a better trade-off between channel power gain and movement time compared to these two benchmarks.

\begin{figure}[!t]
	\centering
	\includegraphics[width=3in]{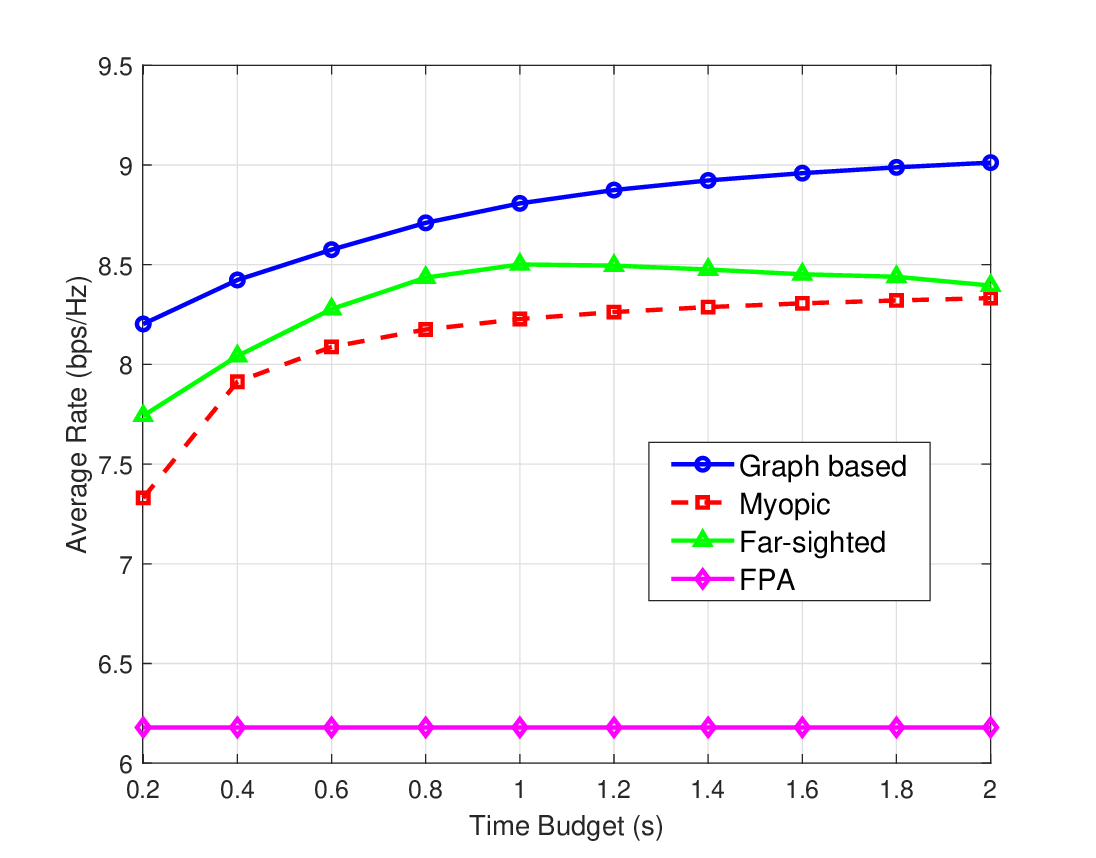}
	\caption{Average rates versus time duration.}
	\label{fig.6}
    \vspace{-9pt}
\end{figure}
In Fig. \ref{fig.6}, we plot the average data rates by different schemes against the time duration $T$. It is observed that the proposed scheme yields better rate performance than other benchmark schemes. It is interesting to note that the rate performance of the far-sighted scheme gradually decreases as $T \geq 1$ s. This is because a longer $T$ allows the MA to move to a more distant high-gain position, and it may experience more significant deep fading during the movement phase. Moreover, the FPA benchmark achieves the worst performance among all schemes considered, as expected.

\begin{figure}[!t]
	\centering
	\includegraphics[width=3in]{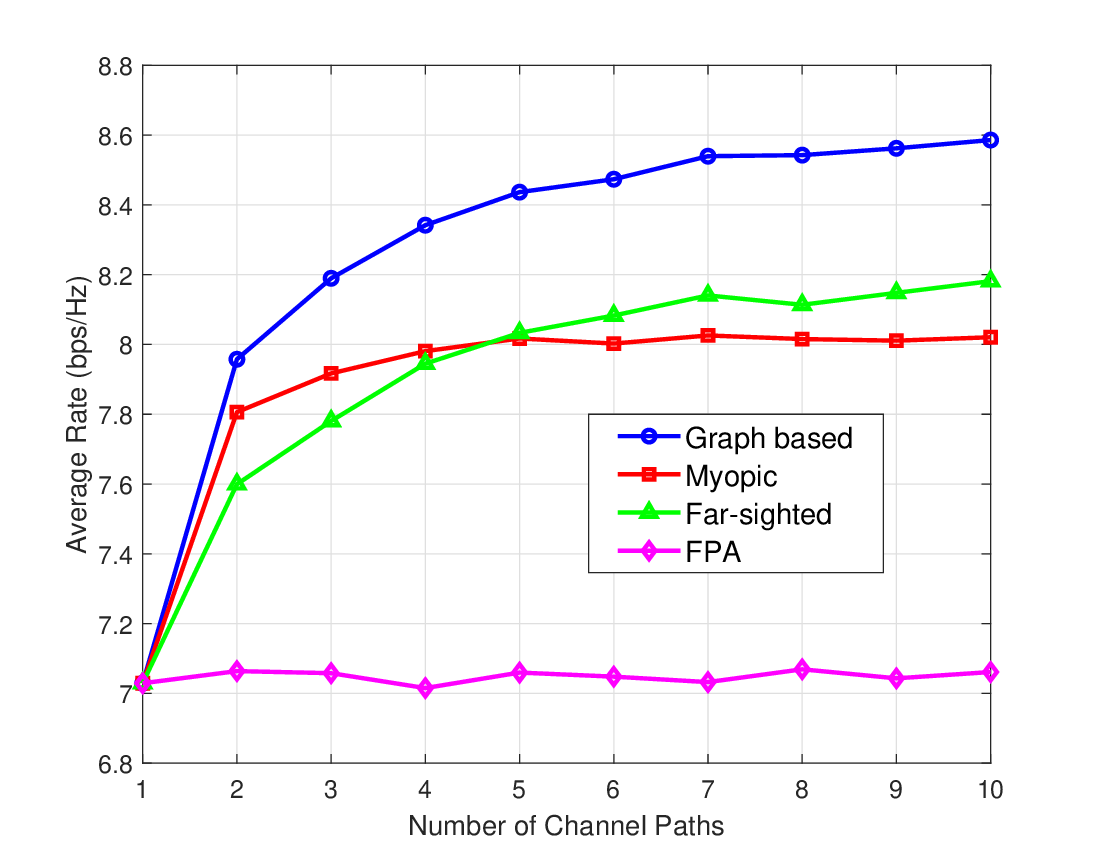}
	\caption{Average rate versus number of channel paths.}
	\label{fig.7}
    \vspace{-9pt}
\end{figure}
Fig. \ref{fig.7} plots the average data rate versus the number of transmit paths $L_t$, with $T=2$ s. It is observed that all schemes with the MA achieve an increasing average data rate as $L_t$ increases, thanks to more significant spatial diversity, while the rate performance of the FPA scheme remains approximately constant. Moreover, the performance gap between the proposed scheme and the myopic and far-sighted schemes is observed to increase with $L_t$. This is because the enhanced spatial diversity provides more flexibility for the proposed scheme to balance the movement time and channel power gain. Nonetheless, all the considered schemes are observed to achieve the same performance for $L_t=1$, since a uniform channel power gain is achieved within the movement region in this case. It is also observed that the far-sighted scheme outperforms the myopic scheme as $L_t \ge 5$. This is because the far-sighted scheme can better exploit the pronounced spatial diversity gain provided by increasing $L_t$ thanks to its larger movement range. Additionally, the sufficiently large time budget ($T=2$ s) helps mitigate the deep fading that may occur during the movement.

\begin{figure}[!t]
	\centering
	\includegraphics[width=3in]{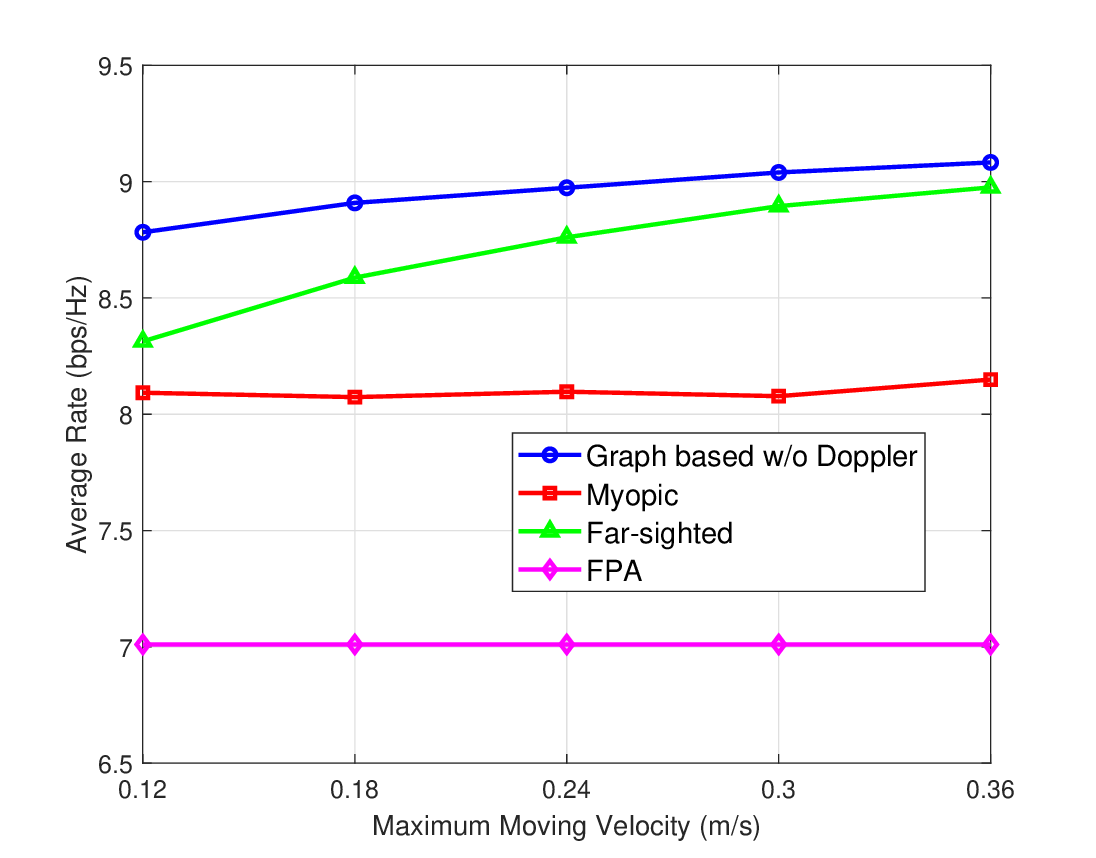}
	\caption{Average rate versus MA maximum velocity.}
	\label{fig.8}
    \vspace{-9pt}
\end{figure}
Finally, Fig. \ref{fig.8} shows the average data rate versus the MA's maximum moving velocity $v_{\max}$, with $T=2$ s. It is observed that the performance of both the proposed scheme and the far-sighted scheme improves as $v_{\max}$ increases. This is because a larger $v_{\max}$ expands the movement range within the given time budget $T$, providing more DoFs in trajectory design. Moreover, a larger $v_{\max}$ enables the MA to reach a desirable destination position more quickly, allowing it to maintain high channel power gain for a longer duration. It is also observed that the performance gap between our proposed scheme and the far-sighted scheme gradually diminishes as $v_{\max}$ increases. This is because a higher $v_{\max}$ reduces the impact of deep fading during movement, allowing the MA to explore more distant destinations with higher channel power gains. In contrast, the performance of the myopic scheme remains nearly unchanged as $v_{\max}$ increases, as its destination position is always the nearest crest, limiting its ability to explore more distant positions.\vspace{-4pt}

\section{Conclusion}
This paper investigated a new joint communication and trajectory optimization problem for a single-MA system. To achieve its optimal performance, we optimize the MA's trajectory to maximize the average data rate over a given time interval at the receiver subject to the MA's maximum velocity constraints. For the special case with two channel paths, the optimal trajectory was derived in closed-form, showing that the MA should always move at the maximum velocity towards the nearest crest. For the general multi-path case, the MA trajectory design problem was reformulated as a fixed-hop shortest path problem in graph theory and solved optimally using DP. Simulation results demonstrated that the proposed design properly balances the moving distance/time and rate performance, thereby significantly outperforming other heuristic and FPA baselines. 
\vspace{-6pt}

\appendix[Proof of Proposition 1]
First, it can be easily shown that if $x_0$ is exactly located at a valley, \textbf{Proposition 1} always holds. Next, we consider the following two cases where $x_0$ is not located at a valley. 

\textbf{Case 1:} $x^*$ can be reached from $x_0$ within the duration $T$, i.e., $\frac{\lvert x^*-x_0 \rvert}{v_{\max}\tau}\leq K$. Let $k_1 = \frac{|x^* - x_0|}{v_{\max}\tau}, k_1\in\mathcal{K},$ be the required number of time slots for the MA to reach $x^*$. Then, the MA trajectory in \eqref{eq01} becomes
\begin{equation}
	x^*[k] = \begin{cases}
		x_0+\frac{x^*-x_0}{\left|x^*-x_0\right|}v_{\max}k\tau, &\ 0 \leq k < k_1,\\
		x^*, &\ k_1 \leq k \leq K.
	\end{cases}\label{eq16}
\end{equation}
Based on \eqref{eq16}, the average rate $\frac{1}{K}\sum_{k=1}^{K}R[k]$ can be decomposed as $\frac{1}{K}\left(\sum_{k=1}^{k_1}R[k]+(K-k_1)R[K] \right)$. Consider an arbitrary feasible trajectory, denoted as $x'[k], k \in {\cal K}$, and let $R'[k]$ denote the achievable rate at the receiver for $x=x'[k]$. Evidently, we have $(K-k_1)R[K] \geq \sum_{k=k_1}^{K}R'[k]$. Next, we show $\sum_{k=1}^{k_1}R[k] \ge \sum_{k=1}^{k_1}R'[k]$. Note that moving the MA from $x_0$ to $x'$ within the first to the $k_1$-th time slot always yields an increasing channel power gain (see the green curve in Fig.\,\ref{fig3}), while moving the MA in the opposite direction always yields a decreasing channel power gain first within this period (see the red curve in Fig.\,\ref{fig3}). Hence, if the trajectory $x'[k], k \in {\cal K}$ moves the MA away from $x^*$, we have $R[k] > R'[k], 1 \le k \le k_1$, and thus, $\sum_{k=1}^{k_1}R[k] \ge \sum_{k=1}^{k_1}R'[k]$.

Thus, the optimal trajectory should always move the MA towards $x^*$. Under this premise, if the trajectory $x'[k], k \in {\cal K}$ moves the MA at a lower velocity than $v_{\max}$, it can also be seen that $R[k] > R'[k]$. By combining the above results, we have $\frac{1}{K}\sum_{k=1}^{K}R[k] \ge \frac{1}{K}\sum_{k=1}^{K}R'[k]$.
\begin{figure}[!t]
	\centering
	\includegraphics[width=3in]{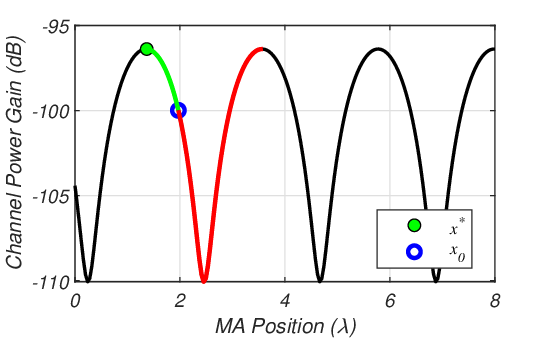} 
	\caption{Illustration of the proof of Proposition 1.}
	\label{fig3}
    \vspace{-9pt}
\end{figure}

\textbf{Case 2:} $x^*$ cannot be reached from $x_0$ within the time duration $T$, i.e., $\frac{\lvert x^*-x_0 \rvert}{v_{\max}\tau} > K$. In this case, we can conduct a similar analysis as in Case 1 to show $\frac{1}{K}\sum_{k=1}^{K}R[k] \ge \frac{1}{K}\sum_{k=1}^{K}R'[k]$. The details are omitted for brevity.

\footnotesize
\bibliographystyle{IEEEtran}
\bibliography{references.bib}
\end{document}